\documentclass[sigconf,nonacm]{acmart}

\AtBeginDocument{%
  }

\setcopyright{acmlicensed}
\copyrightyear{2025}
\acmYear{2025}
\acmDOI{XXXXXXX.XXXXXXX}
\acmConference[NANOCOM '25]{The 12th Annual ACM International Conference on Nanoscale Computing and Communication}{Oct. 23--25,
  2025}{Chengdu, China}

\acmISBN{978-1-4503-XXXX-X/2018/06}

\usepackage{amsmath}

\usepackage{amssymb}

\usepackage{acronym}
\usepackage{multirow}
\usepackage{algorithm}
\usepackage{algpseudocode}

\newcommand{\scaleSection}{\vspace*{-0.25cm}}
\newcommand{\scaleSubsection}{\vspace*{-0.25cm}}
\newcommand{\scaleSectionBelow}{\vspace*{-0.1cm}}
\newcommand{\scaleSubsectionBelow}{\vspace*{-0.1cm}}
\newcommand{\scaleProof}{\vspace*{-2.5mm}}
\newcommand{\scaleProposition}{\vspace*{-1.5mm}}

\acrodef{MC}{molecular communication}
\acrodef{ML}{maximum likelihood}
\acrodef{AML}{approximate maximum likelihood}
\acrodef{EM}{electro-magnetic}
\acrodef{MOS}{metal-oxide semi-conductor}
\acrodef{CSK}{concentration shift keying}
\acrodef{RSK}{ratio shift keying}
\acrodef{MSK}{molecule shift keying}
\acrodef{GMoSK}{generalized molecule shift keying}
\acrodef{BCSK}{binary concentration shift keying}
\acrodef{MMSK}{molecule mixture shift keying}
\acrodef{MIMO}{multiple-input multiple-output}
\acrodef{UT}{unscented transform}
\acrodef{RX}{receiver}
\acrodef{TX}{transmitter}
\acrodef{RV}{random variable}
\acrodef{pdf}{probability density function}
\acrodef{ISI}{inter-symbol interference}
\acrodef{SIN}{signal-independent noise}
\acrodef{SDCN}{signal-dependent channel noise}
\acrodef{BER}{bit error rate}
\acrodef{SER}{symbol error rate}
\acrodef{SNR}{signal-to-noise ratio}
\acrodef{kNN}{$k$-nearest-neighbor}
\acrodef{PEP}{pairwise error probability}
\acrodef{MDA}{mixture design algorithm}

\acrodef{SOD}{sensor output domain}
\acrodef{SID}{sensor input domain}

\renewcommand{\a}{\mathbf{a}}
\renewcommand{\b}{\mathbf{b}}
\newcommand{\sensorone}{^\mathrm{TGS800}}
\newcommand{\sensortwo}{^\mathrm{TGS826}}
\renewcommand{\u}{\mathbf{u}}
\newcommand{\x}{\mathbf{x}}

\newcommand{\xbar}{\mathbf{\bar{x}}}
\newcommand{\xhat}{\mathbf{\hat{x}}}
\renewcommand{\H}{\mathbf{H}}
\newcommand{\h}{\mathbf{h}}
\newcommand{\y}{\mathbf{y}}
\newcommand{\z}{\mathbf{z}}
\newcommand{\nbase}{\mathbf{n}^{\mathrm{norm}}}
\newcommand{\ntx}{\mathbf{n}^{\mathrm{TX}}}
\newcommand{\nc}{\mathbf{n}^{\mathrm{C}}}
\newcommand{\nrx}{\mathbf{n}^{\mathrm{RX}}}
\newcommand{\f}[1]{\mathbf{f}\left( #1 \right)}
\newcommand{\fMOS}[1]{\mathbf{f}\MOS\left( #1 \right)}

\newcommand{\g}[1]{\mathbf{g}\left( #1 \right)}
\newcommand{\nuc}{\nu^{\mathrm{C}}}

\newcommand{\E}[1]{\mathrm{E}\left\{ #1 \right\}}
\newcommand{\Cov}[1]{\mathrm{Cov}\left\{ #1 \right\}}

\newcommand{\diagvec}[1]{\mathrm{diag}\left\{ #1 \right\}}
\newcommand{\diagmat}[1]{\mathrm{Diag}\left\{ #1 \right\}}

\newcommand{\symbolalphabet}{\mathcal{X}}
\newcommand{\feasibleset}{\mathcal{X}'}

\newcommand{\nsymbols}{N}
\newcommand{\nspecies}{S}
\newcommand{\nsensors}{R}
\newcommand{\nsigmapoints}{n_{\sigma}}

\newcommand{\I}{\mathbf{I}}
\newcommand{\nullmatrix}{\mathbf{0}}

\newcommand{\muvec}{\boldsymbol{\mu}}
\newcommand{\covmat}{\mathbf{C}}

\newcommand{\muvecapprox}{\boldsymbol{\Tilde{\mu}}}
\newcommand{\covmatapprox}{\mathbf{\Tilde{C}}}

\newcommand{\SIN}{^{\mathrm{SIN}}}

\newcommand{\MOS}{^{\mathrm{MOS}}}
\newcommand{\ML}{^{\mathrm{ML}}}

\newcommand{\deuclidean}{d^{\ell_2}}
\newcommand{\dvahid}{d^{\mathrm{SNR}}}
\newcommand{\dber}{d^{\mathrm{PEP}}}

\newcommand{\s}{\mathbf{s}}

\newcommand{\transpose}{^\mathrm{T}}
\newcommand{\hadamard}{\odot}

\newcommand{\covfunc}{\mathbf{c}}
\newcommand{\alphavec}{\boldsymbol{\alpha}}
\newcommand{\betavec}{\boldsymbol{\beta}}

\newtheoremstyle{gen}
  {\topsep}    %
  {\topsep}    %
  {}  %
  {}           %
  {\bfseries}  %
  {.}          %
  { }          %
  {Proposition (#3, General Case)} %
\theoremstyle{gen}
\newtheorem*{resultgeneral}{Proposition}

\newtheoremstyle{sin}
  {\topsep}    %
  {\topsep}    %
  {}  %
  {}           %
  {\bfseries}  %
  {.}          %
  { }          %
  {Proposition (#3, \ac{SIN})} %
\theoremstyle{sin}
\newtheorem*{resultsin}{Proposition}

\newtheoremstyle{sdcn}
  {\topsep}    %
  {\topsep}    %
  {}  %
  {}           %
  {\bfseries}  %
  {.}          %
  { }          %
  {Proposition (#3, \ac{SDCN})} %
\theoremstyle{sdcn}
\newtheorem*{resultsdcn}{Proposition}

\newtheoremstyle{sinsdcn}
  {\topsep}    %
  {\topsep}    %
  {}  %
  {}           %
  {\bfseries}  %
  {.}          %
  { }          %
  {Proposition (#3, \ac{SIN} and \ac{SDCN})} %
\theoremstyle{sinsdcn}
\newtheorem*{resultsinsdcn}{Proposition}

\begin{document}

\title[Molecule Mixture Detection and Alphabet Design]{Molecule Mixture Detection and Alphabet Design for Non-linear, Cross-reactive Receiver Arrays in MC} %

\author{Bastian Heinlein$^{1,2}$, Kaikai Zhu$^{1}$, Sümeyye Carkit-Yilmaz$^{1}$, Sebastian Lotter$^1$, Helene M. Loos$^{1,3,4}$, \mbox{Andrea Buettner$^{1,4}$,} Yansha Deng$^5$, Robert Schober$^1$, and Vahid Jamali$^2$}
\affiliation{%
\institution{$^1$ Friedrich-Alexander-Universität Erlangen-Nürnberg, Erlangen, Germany}\country{}
}
\affiliation{%
\institution{$^2$ Technical University of Darmstadt, Darmstadt, Germany}\country{}
}
\affiliation{%
\institution{$^3$ Rheinische Friedrich-Wilhelms-Universität Bonn, Bonn, Germany}\country{}
}
\affiliation{%
\institution{$^4$ Fraunhofer Institute for Process Engineering and Packaging IVV, Freising, Germany}\country{}
}
\affiliation{%
\institution{$^5$ Department of Engineering, King’s College London, London, UK}\country{}
}
\email{{bastian.heinlein,kaikai.zhu,suemeyye.yilmaz,sebastian.g.lotter,helene.loos,andrea.buettner,robert.schober}@fau.de}
\email{yansha.deng@kcl.ac.uk, vahid.jamali@rcs.tu-darmstadt.de}

\renewcommand{\shortauthors}{Heinlein et al.}

\begin{abstract}
    {\setlength{\baselineskip}{0pt}
    Air-based \ac{MC} has the potential to be one of the first \ac{MC} systems to be deployed in real-world applications, enabled by existing sensor technologies such as \ac{MOS} sensors. However, commercially available sensors usually exhibit non-linear and cross-reactive behavior, contrary to the idealizing assumptions about linear and perfectly molecule type-specific sensing often made in the \ac{MC} literature. 
    To address this gap, we propose a detector for molecule mixture communication with a general non-linear, cross-reactive \ac{RX} array that performs approximate maximum likelihood detection on the sensor outputs. Additionally, we introduce an algorithm for the design of mixture alphabets that accounts for the \ac{RX} characteristics.
    We evaluate our detector and alphabet design algorithm through simulations that are based on measurements reported for two commercial \ac{MOS} sensors. Our simulations demonstrate that the proposed detector achieves similar symbol error rates as data-driven methods without requiring large numbers of training samples and that the alphabet design algorithm outperforms methods that do not account for the \ac{RX} characteristics. 
    Since the proposed detector and alphabet design algorithm are also applicable to other chemical sensors, they pave the way for reliable \mbox{air-based \ac{MC}}.}
    \vspace*{-3mm}
\end{abstract}

\keywords{\vspace*{-1mm}Air-based Molecular Communication, Molecule Mixtures}

{\setlength{\baselineskip}{8pt}\maketitle}

\acresetall
{
\setlength{\baselineskip}{10pt}
\scaleSection\section{Introduction}\scaleSectionBelow\label{sec:introduction}
\Ac{MC} is a new paradigm employing molecules for information transmission in scenarios where \ac{EM} wave-based communication is not practical such as for interfacing with biological entities via chemical signals or in \textit{\ac{EM} wave-denied} environments \cite{guo:molecular_physical_layer_6g}. Thanks to existing \ac{TX} hardware (e.g., odor delivery devices) and sensor technology (e.g., \ac{MOS} sensors), air-based \ac{MC} in particular has the potential to become practically feasible in the near future. 
However, while most existing work in \ac{MC} relies on the assumption of molecule counting \acp{RX}, i.e., a linear relationship between the response $r$ of the sensing unit\footnote{We use the term \textit{sensing unit} to emphasize that the investigated effects are present in both natural systems, e.g., biological receptors, and synthetic systems, e.g., \ac{MOS} sensors, even though we focus on the latter systems in this work.} and the input concentration $c$ \cite{kuscu:tx_rx_architectures_mc}, both synthetic and natural sensing units usually exhibit \textbf{non-linear behavior}, i.e., $r = f(c)$ with some non-linear function $f(\cdot)$. 
For example, power-law behavior, i.e., $f(c) = a \cdot c^b$, $a,b \in \mathbb{R}$, is ubiquitous and can be observed in both synthetic systems like \ac{MOS} sensors~\cite{madrolle:linear_quadratic_model_quantification_mixture_two_diluted_gases_single_MOS} and in natural systems, e.g., for the sense of smell~\cite{teixeira:perception_fragrance_mixtures}.
Even sensing units that have a linear regime exhibit some non-linear behavior due to saturation at high concentrations \cite{kuscu:tx_rx_architectures_mc}.  
As an additional challenge, sensors and biological receptive structures, e.g., odorant receptors, are often \textbf{cross-reactive}, i.e., their response depends on the concentrations of several molecule types in the environment, albeit with different response intensity for different molecule types \cite{madrolle:linear_quadratic_model_quantification_mixture_two_diluted_gases_single_MOS}. Yet, in the \ac{MC} literature, cross-reactive behavior is often not properly accounted for and instead most \ac{MC} works consider either only one molecule type or perfectly orthogonal sensing units~\cite{jamali:olfaction_inspired_MC}.

In practice, sensing units exhibit both non-linear and cross-reactive behavior at the same time, i.e., the response of a sensor is generally described by a non-linear function $f(c_1, \dots, c_{\nspecies})$ that depends on the concentrations $c_i$, $i \in \{1, \dots, \nspecies \}$, of $\nspecies$ molecule types in the environment. 
Therefore, these two characteristics have to be considered when developing detection algorithms and designing signal constellations for practical \ac{MC} systems. This is of particular relevance for molecule mixture modulation, where the information is encoded onto the concentration of multiple molecule types, promising higher information content per symbol~\cite{jamali:olfaction_inspired_MC} and higher robustness against changes in the environment~\cite{araz:ratio_shift_keying_time_varying_mc_channels} compared to conventional modulation schemes such as \ac{CSK}. 

This paper makes two main contributions for molecule mixture communication with \ac{RX} arrays that exhibit non-linear, cross-reactive behavior\footnote{It is worth noting that molecule mixture communication includes as special cases some well-known simpler schemes such as \ac{CSK}.}: 1) We propose a \textbf{suboptimal low-complexity detector} based on first- and second-order moments applicable to general non-linear, cross-reactive \ac{RX} arrays and molecule mixture communications. Thereby, we fill a gap in the \ac{MC} literature where only relatively simple special cases of non-linear or cross-reactive \acp{RX}~\cite{kim:experimentally_validated_channel_model_for_mc_systems,bhattacharjee:detection_process_macroscopic_airbased_mc_using_PIN_photodiode,wietfeld:evaluation_multi_molecule_molecular_communication_testbed,araz:ratio_shift_keying_time_varying_mc_channels,jamali:olfaction_inspired_MC} have been considered but no general approach to address this challenge has been reported, especially in the context of molecule mixture communication. 2) We further propose a \textbf{mixture alphabet design algorithm} that selects mixtures to achieve good separability of the outputs of the sensing units. This is different from most existing modulation schemes, where usually either 0 or $N_{\mathrm{released}}$ molecules of each molecule type are released, with $N_{\mathrm{released}}$ chosen ad-hoc (see~\cite{kuren:survey_modulation_techniques_MCD} for an overview). 
While \cite{wang:effective_constellation_design_CSK} and \cite{jamali:olfaction_inspired_MC} move beyond these simple constellation designs by respectively optimizing the signal constellation of a higher-order \ac{CSK} modulation and designing a \ac{MMSK} alphabet, both make simplifying assumptions regarding the \ac{RX} characteristics (\cite{wang:effective_constellation_design_CSK} assumes a molecule counting \ac{RX} and \cite{jamali:olfaction_inspired_MC} only optimizes for a specific non-linear, cross-reactive \ac{RX} implementation). In contrast, our alphabet design algorithm is applicable to mixture communication with arbitrary \ac{RX} characteristics. 
Both the proposed detector and the alphabet design algorithm are designed for a general single-sample \ac{RX} with minimal assumptions regarding the \ac{TX}, channel, and \ac{RX}. Besides a general formulation, we provide direct solutions for two special cases, namely, when i) \ac{TX}, channel, and \ac{RX} noise are signal-independent, and ii) when the channel noise is signal-dependent while the \ac{TX} and \ac{RX} noise remain signal-independent.

\textbf{Notation:} Vectors and matrices are denoted by lowercase and uppercase bold letters, respectively. The $i$-th entry of vector $\x$ is denoted by $[\x]_i$ and the $j$-th column of matrix $\covmat$ as $[\covmat]_j$. The transpose of matrix $\covmat$ is denoted by $\covmat\transpose$ whereas its determinant is denoted by $\det\{\covmat\}$. $\diagvec{\x}$ denotes a diagonal matrix having the elements of vector $\x$ on its main diagonal, while $\diagmat{\covmat}$ denotes a column vector containing the diagonal elements of matrix $\covmat$.
Element-wise multiplication of two vectors $\x$ and $\y$ is shown by $\x\hadamard\y$. The \ac{pdf} of a \ac{RV} $x$ conditioned on $y$ is given by $p_{x|y}(x|y)$. The expectation of a \ac{RV} is denoted by $\E{\cdot}$ and the covariance of a random vector by $\Cov{\cdot}$. $\I$ denotes the identity matrix and $\nullmatrix$, depending on the context, either a vector or a matrix containing only zeroes. $\lVert \cdot \rVert_2$ denotes the $\ell_2$-norm. 
Sets are shown by calligraphic letters and the cardinality of set $\mathcal{A}$ is denoted by $|\mathcal{A}|$. $\mathcal{U}(\mathcal{A})$ denotes a \ac{RV} that is uniformly distributed over a \mbox{set $\mathcal{A}$}. 

\scaleSection\section{System Model}\scaleSectionBelow\label{sec:system_model}
In the following, we first introduce the system model in its general form in Section~\ref{sec:system_model:overview} before describing its relevant special cases in Section~\ref{sec:system_model:special_cases}.

\scaleSubsection\subsection{General Description}\scaleSubsectionBelow\label{sec:system_model:overview}
We consider a general system model that incorporates the three major noise sources reported in the experimental \ac{MC} literature, namely \ac{TX} noise, i.e., uncertainty about the amount of released molecules \cite{farsad:tabletop_mc_text_messages_through_chemical_signals,kim:experimentally_validated_channel_model_for_mc_systems}, channel noise, i.e., uncertainty due to the propagation of the molecules and potential interference \cite{kim:experimentally_validated_channel_model_for_mc_systems}, and \ac{RX} noise, i.e., noise introduced by the employed sensing units at the \ac{RX} \cite{shin:low_frequency_noise_gas_sensors,farsad:tabletop_mc_text_messages_through_chemical_signals}. 
We consider here a single-sample detector to focus on the effect of the \ac{RX} array characteristics instead of the intricate temporal dynamics that proved challenging to model even for special cases (see, e.g., \cite{khaloopour:experimental_platform_macroscale_fluidic,alzubi:macroscale_MC_iot_based_pipeline_inspection}), let alone for a general setting.
After the simplification of neglecting the effect of \ac{ISI} to focus on the effect of the \ac{RX} array characteristics, we arrive at the following system model\footnote{A common approach in the \ac{MC} literature is to assume sufficiently long symbol intervals so that \ac{ISI} can be neglected \cite{kilinc:rx_design_for_mc}. As the information content per symbol can be increased by employing molecule mixture modulation as compared to schemes like \ac{CSK}, longer symbol intervals can be chosen to achieve a desired data rate, thus making this assumption less limiting for molecule mixture communications. In addition, methods like equalization can be employed to mitigate \ac{ISI}~\cite{kilinc:rx_design_for_mc}.}, whose components we will define in the next paragraphs:
\begin{subequations}
    \begin{align}
        \x &= \xbar +\ntx(\xbar) \label{eq:system_model:tx_noise} \\
        \y &= \H \x + \nc(\x) \label{eq:system_model:channel_noise}\\
        \z &= \f{\y} + \nrx(\f\y).\label{eq:system_model:rx_noise}
    \end{align}
    \end{subequations}
Equs.~\eqref{eq:system_model:tx_noise}, \eqref{eq:system_model:channel_noise}, and \eqref{eq:system_model:rx_noise} describe respectively the noisy release process of the signaling molecules, their propagation through a noisy channel, and how the molecule concentrations at the \ac{RX} are converted to the eventual measurements that can be used for symbol detection. 

In \eqref{eq:system_model:tx_noise}, $\x \in \mathbb{R}_{\geq0}^\nspecies$, where $\mathbb{R}_{\geq0}$ denotes the set of non-negative real numbers, is the vector of the concentrations of all $\nspecies$ molecule types immediately after the molecules have been released from the \ac{TX}. For an ideal \ac{TX}, $\x=\xbar$ holds, where $\xbar$ collects the desired molecule concentrations of a symbol after \textbf{molecule release}. However, practical \acp{TX} cannot perfectly control the amounts of molecules they release~\cite{farsad:tabletop_mc_text_messages_through_chemical_signals}. We model this by additive noise $\ntx(\xbar)$, which can in general be symbol-dependent, and follows the conditional \ac{pdf} $p_{\ntx| \xbar}(\ntx|\xbar)$. 
In practice, the amount of released molecules cannot be arbitrary, but must consider the operating regime of the \ac{RX}, limited molecule storage, and hardware constraints limiting the amounts of molecules of each type that can be released at once. Therefore, we require that all symbols $\xbar \in \symbolalphabet$ in the mixture alphabet $\symbolalphabet$ lie in a \textit{feasible set} $\feasibleset \subseteq \mathbb{R}_{\geq0}^\nspecies$.

Eq.~\eqref{eq:system_model:channel_noise} describes the molecule concentrations $\y \in \mathbb{R}_{\geq0}^\nspecies$ at the passive \ac{RX} after \textbf{propagation through the noisy channel}. Depending on the environmental conditions, such as airflow and the distance between \ac{TX} and \ac{RX}, $\y$ may be subject to attenuation in the channel, captured by diagonal channel matrix $\H = \diagvec{\h}$, where $[\h]_i$ denotes the attenuation factor for molecules of type $i$. The channel noise $\nc(\x)$ can in principle be signal-dependent and is distributed according to a conditional \ac{pdf} $p_{\nc|\x}(\nc|\x)$. 

Finally, \eqref{eq:system_model:rx_noise} describes how the molecule concentrations $\y$ are converted to the \textbf{response of the sensing units} $\z \in \mathbb{R}^\nsensors$, where $[\z]_r$, $r \in \{1, \dots, \nsensors\}$, denotes the output of the $r$-th sensing unit and $\nsensors$ the total number of sensing units at the \ac{RX}. The mapping from $\y$ to $\z$ follows in general a non-linear function $\f\cdot$, whose output depends on multiple molecule concentrations simultaneously. 
In practice, the sensing units also introduce measurement noise $\nrx(\f\y)$, e.g., thermal noise. For the sake of generality, we again permit signal-dependent measurement noise and characterize it by conditional \ac{pdf} $p_{\nrx|\f\y}(\nrx|\f\y)$. 

In summary, we consider a very general system model that accounts for the three relevant noise sources (\ac{TX}, channel, and \ac{RX} noise) reported in experimental \ac{MC} literature and a possibly non-linear sensor function. 

\scaleSubsection\subsection{Relevant Special Cases}\scaleSubsectionBelow\label{sec:system_model:special_cases}
Next, we specialize the system model proposed in Section~\ref{sec:system_model:overview} to relevant special cases that are commonly encountered in \ac{MC}. 

\subsubsection{Signal-independent Gaussian Noise}\label{sec:system_model:special_cases:SIN}
The first special case we consider is that all three noise sources are signal-independent and Gaussian. Signal-independent \ac{TX} noise, characterized by mean $\muvec_{\ntx}\SIN$ and covariance $\covmat_{\ntx}\SIN$, might occur, e.g., due to small errors when controlling the flow rate of an odor delivery device \cite{hopper:multichannel_portable_odor_delivery_device} and signal-independent \ac{RX} noise, characterized by mean $\muvec_{\nrx}\SIN$ and covariance $\covmat_{\nrx}\SIN$, might be due to thermal noise in the electronic components \cite{shin:low_frequency_noise_gas_sensors}. For high molecule concentrations and controlled propagation environments, the observed number of molecules at the \ac{RX} converges to its expectation (see~\cite[Sec. 2.2.2.4]{jamali:thesis}).
The channel noise, characterized by mean $\muvec_{\nc}\SIN$ and covariance $\covmat_{\nc}\SIN$, might be used to represent the uncertainty about the amount of background molecules in the environment, e.g., due to unknown interference which is independent of the currently transmitted signal. 
We consider Gaussian noise for simplicity since it has been shown to be an accurate approximation in air-based \ac{MC} \cite{kim:experimentally_validated_channel_model_for_mc_systems,mcguiness:experimental_results_openair_transmission_macromolecular_communication} and since Gaussian \acp{RV} are readily parametrized by their mean and covariance. Yet, the proposed detector and mixture alphabet design algorithm are also applicable to other noise distributions, although possibly with some performance degradation since they exploit only mean and covariance. %

\subsubsection{Signal-dependent Channel Noise}\label{sec:system_model:special_cases:SDCN}
In addition to the case of all-\ac{SIN}, we consider the more challenging scenario of \ac{SDCN}, while the \ac{TX} and \ac{RX} noise remain signal-independent and Gaussian (i.e., they are described by $\muvec_{\ntx}\SIN$ and $\covmat_{\ntx}\SIN$, $\muvec_{\nrx}\SIN$ and $\covmat_{\nrx}\SIN$ as in Section~\ref{sec:system_model:special_cases:SIN}). Signal-dependent channel noise $\nc(\x)$, which occurs, e.g., in Poisson channels, is commonly assumed in \ac{MC} \cite{jamali:olfaction_inspired_MC} and often approximated as
\vspace*{-2.5mm}\begin{equation}\label{eq:system_model:SDCN:noise_definition}
    \nc(\x) = \sqrt{\nuc\H\x} \hadamard \nbase,
\vspace*{-1.5mm}\end{equation}%
where square root $\sqrt{\cdot}$ is applied element-wise and $\nbase$ is normal distributed with zero-mean and covariance matrix $\I$. We introduce a scaling factor $\nuc$ to be able to control the noise power introduced by the channel in our simulations, with $\nuc=1$ corresponding to the Gaussian approximation of the Poisson channel. 

\subsubsection{MOS Sensors as Non-Linear, Cross-Reactive RX Arrays}\label{sec:system_model:special_cases:MOS}
Finally, we consider an array of two \ac{MOS} sensors at the \ac{RX} as an example of a non-linear sensing unit. As already discussed in Section~\ref{sec:introduction}, \ac{MOS} sensors exhibit power-law behavior, i.e., their conductance for a concentration $c$ of a target gas can be described by $G = a c^b$, where $a,b \in \mathbb{R}$ are constants specific to the considered combination of sensor and target gas \cite{madrolle:linear_quadratic_model_quantification_mixture_two_diluted_gases_single_MOS}.
While the response of \ac{MOS} sensors to gas mixtures is less explored, some models also have been proposed to describe the response to mixtures of two gases (see, e.g., \cite{madrolle:linear_quadratic_model_quantification_mixture_two_diluted_gases_single_MOS} for a comparison of four existing models).  
Here, we focus on the model originally introduced in \cite{llobet:steadystate_transient_behavior_of_thickfilm_tin_oxide_sensors_gas_mixtures} since it is well-established in the literature, can be readily extended to more than two gases (see~\cite{llobet:steadystate_transient_behavior_of_thickfilm_tin_oxide_sensors_gas_mixtures}), and showed a good fit to measurement data in \cite{madrolle:linear_quadratic_model_quantification_mixture_two_diluted_gases_single_MOS}. The model is given as 
\vspace*{-1.5mm}\begin{equation}
    G(c_1, c_2) = a_1 c_1^{b_1} + a_2 c_2^{b_2} + a_3 c_1^{b_1} c_2^{b_2} + a_4,\vspace*{-0.5mm}
\end{equation}
where $a_1, a_2, a_3, a_4, b_1, b_2$ are constants that are collected in vectors $\a = [a_1, a_2, a_3, a_4]$ and $\b = [b_1, b_2]$, and $c_1$ and $c_2$ denote the concentrations of the respective target gases. 
In the notation of our system model, the function describing the \ac{RX} array is thus given by
\vspace*{-2mm}\begin{equation}\label{eq:system_model:mos}
    \fMOS{y_1, y_2} = \begin{bmatrix} f^1(y_1, y_2) \\ f^2(y_1, y_2) \end{bmatrix} = \begin{bmatrix}
        a_1^1 y_1^{b_1} + a_2^1 y_2^{b_2^1} + a_3^1 y_1^{b_1^1} y_2^{b_2^1} + a_4^1 \\
        a_1^2 y_1^{b_1^2} + a_2^2 y_2^{b_2^2} + a_3^2 y_1^{b_1^2} y_2^{b_2^2} + a_4^2 \\
    \end{bmatrix},
\end{equation}
where $f^r(y_1, y_2)$, $a_q^r$, and $b_s^r$ denote the output of the $r$-th sensor, the coefficients of the $r$-th sensor with $q \in \{1, \dots, 4\}$, and the power-law associated with the $r$-th sensor and molecules of type $s$, respectively. 
We emphasize again that our detector and mixture alphabet design algorithm are applicable for more than two molecule types and for arbitrary functions $\f\cdot$. 

\scaleSection\section{Mixture Detection Using First- and Second-Order Moments}\scaleSectionBelow\label{sec:detection}
In this section, we present our first main contribution, namely an \ac{AML} detector based on first- and second-order moments. To this end, we provide an overview over the proposed detector, briefly review the \ac{UT}, which is essential to deal with the non-linearity introduced by the sensing units, and derive the approximate likelihoods of the sensing unit responses for each symbol.

\scaleSubsection\subsection{Approximate ML Detection}\scaleSubsectionBelow\label{sec:detection:overview}
Since the proposed detector is inspired by \ac{ML} detection, we briefly review the optimal \ac{ML} detector before outlining the proposed suboptimal \ac{AML} detector. 
In \textbf{\ac{ML} detection}, one solves the \textit{forward problem}, i.e., determining the likelihoods of each symbol and then using them to estimate which symbol $\xbar \in \symbolalphabet$ has been transmitted. For our system model, \ac{ML} detection would be performed according to
\vspace*{-1mm}\begin{equation}
    \xhat\ML = \arg\max_{\xbar \in \symbolalphabet} p_{\z|\xbar}(\z|\xbar),
\vspace*{-2mm}\end{equation}
where $p_{\z|\xbar}(\z|\xbar)$ denotes the \ac{pdf} of $\z$ in the \textit{sensing unit output domain} for transmitted symbol $\xbar$. \ac{ML} detection not only yields optimal performance in terms of \ac{SER} for equi-probable symbols $\xbar \in \symbolalphabet$ but also does not require inversion of $\f\cdot$, which may not be invertible in the first place. 
Unfortunately, obtaining the likelihoods of the sensing unit responses is tedious in practice for non-linear, cross-reactive \ac{RX} arrays. 

Thus, we instead propose an approach for \textbf{\ac{AML} detection} based on the approximate first- and second-order moments of the likelihoods of the sensing unit output for each symbol. In contrast to \ac{ML} detection, we do not determine the exact \ac{pdf} $p_{\z|\xbar}(\z|\xbar)$ but use the corresponding estimated moments $\muvecapprox_{\mathbf{z}|\xbar}$ and $\covmatapprox_{\mathbf{z}|\xbar}$ to parametrize a \ac{pdf} $p^{\mathrm{A}}_{\z|\xbar}(\z|\xbar)$ that approximates $p_{\z|\xbar}(\z|\xbar)$ in the sensing unit output domain, so that \ac{AML} detection can be performed \mbox{according to}:
\vspace*{-2.5mm}\begin{equation}
    \xhat = \arg\max_{\xbar \in \symbolalphabet} p^{\mathrm{A}}_{\z|\xbar}(\z|\xbar).\vspace*{-1mm}
\end{equation}
We focus on first- and second-order moments because they are easily propagated through linear systems and the \ac{UT} provides an efficient way to approximately propagate them through non-linear functions (see Section~\ref{sec:detection:ut}). We use Gaussian distributions, i.e., $p^{\mathrm{A}}_{\z|\xbar}(\z) = \exp[-\frac{1}{2}(\z-\muvecapprox_{\z|\xbar})\transpose\covmatapprox_{\mathbf{z}|\xbar}^{-1}(\z-\muvecapprox_{\z|\xbar})] [(2\pi)^\nsensors \det\{\covmatapprox_{\mathbf{z}|\xbar}\}]^{-1/2}$, since these are readily parametrized by first- and second-order moments and fit well to experiments \cite{mcguiness:experimental_results_openair_transmission_macromolecular_communication}.

It is worth pointing out that our proposed detector is the first one to be applicable to \acp{RX} with arbitrary non-linear, cross-reactive characteristics and requires only knowledge of channel matrix $\H$, the conditional distributions of the \ac{TX}, channel, and \ac{RX} noise, and a way to obtain $\f\y$ for a given $\y$. For a given mixture alphabet $\symbolalphabet$, it is also efficient in terms of computational complexity: During the \textbf{online phase}, i.e., for estimating the symbols, it is only necessary to evaluate the approximate likelihood of each of the $|\symbolalphabet|$ symbols for a given observation $\z$. For the choice of Gaussian \acp{pdf}, this reduces to computing matrix-vector multiplications. During the \textbf{offline phase}, i.e., for obtaining the approximate likelihoods, it is in general necessary to integrate over \acp{pdf}, but we will show that this can be avoided for relevant special cases. 

\scaleSubsection\subsection{Background on the Unscented Transform}\scaleSubsectionBelow\label{sec:detection:ut}
As already briefly mentioned, we employ the \ac{UT} to estimate the mean and covariance of an $\nsensors$-dimensional random vector $\z$ which is obtained by feeding an $\nspecies$-dimensional random vector $\y$ with known mean and covariance through a non-linear function $\f\cdot$, i.e., $\z = \f\y$. 
The main idea behind the \ac{UT}, as originally introduced for state estimation \cite{julier:unscented_filtering_nonlinear_estimation}, is relatively simple: One chooses $\nsigmapoints$ carefully selected points $\s_i$, $i=1, \dots, \nsigmapoints$, so-called \textbf{sigma points} which match the mean $\muvec_\y$ and covariance $\covmat_\y$ of $\y$. These points $\s_i$ are then fed through $\f\cdot$ to obtain $\f{\s_i}$. Then, the mean $\muvec_\z$ and covariance $\covmat_\z$ of $\z$ are estimated respectively as $\muvecapprox_\z = \frac{1}{\nsigmapoints}\sum_{i=1}^{\nsigmapoints} \f{\s_i}$ and $\covmatapprox_{\z} = \frac{1}{\nsigmapoints}\sum_{i=1}^{\nsigmapoints} \left( \f{\s_i}-\muvecapprox_\z \right)\left( \f{\s_i}-\muvecapprox_\z \right)\transpose$.
We follow the approach in \cite{julier:unscented_filtering_nonlinear_estimation} and select the $\nsigmapoints = 2\nspecies$ sigma points  according to $\s_j = \muvec_y + [ (\nspecies \covmat_\y )^{1/2} ]_j$ and $\s_{j+\nspecies} = \muvec_y - [ (\nspecies \covmat_\y)^{1/2} ]_j$, where $j = 1, \dots, \nspecies$, and $\left(\nspecies \covmat_\y\right)^{1/2}$ is obtained using the Cholesky decomposition of $\nspecies \covmat_\y$.

\scaleSubsection\subsection{Derivation of the Moments}\scaleSubsectionBelow\label{sec:detection:moment_derivation}
In this section, we derive the conditional moments that eventually parametrize $p^{\mathrm{A}}_{\z|\xbar}(\z|\xbar)$ for both our general system model and the two special cases of \ac{SIN} and \ac{SDCN}. To this end, we propagate the moments step by step through the system model described by \eqref{eq:system_model:tx_noise}-\eqref{eq:system_model:rx_noise}.

\subsubsection{Moments of $\x$}\label{sec:detection:moments:x}
\scaleProposition
\begin{resultgeneral}[$\x$]
The mean and covariance of $\x$ for a given symbol $\xbar$ can in general be expressed as

\vspace*{-3mm}
\begin{subequations}
        \begin{equation}
            \muvec_{\x|\xbar} = \xbar + \int \ntx p_{\ntx|\xbar}(\ntx|\xbar) \mathrm{d}\ntx
        \end{equation}
        \begin{equation}
            \covmat_{\x|\xbar} = \int \covfunc(\ntx,\muvec_{\ntx|\xbar})p_{\ntx|\xbar}(\ntx|\xbar) \mathrm{d}\ntx,
        \end{equation}
    \end{subequations}

    \vspace*{-1.5mm}
    where $\covfunc(\alphavec, \betavec) = \left(\alphavec - \betavec \right)\left(\alphavec - \betavec \right)\transpose$ with random vector $\alphavec$ and vector $\betavec$, and $\muvec_{\ntx|\xbar}$ is the conditional mean of $\ntx$ for a given $\xbar$.
\end{resultgeneral}
\begin{proof}
    For $\muvec_{\x|\xbar}$, we apply the expectation operator to~\eqref{eq:system_model:tx_noise} and then insert the definition of the first-order moment ($\E{x} = \int x p_{x}(x)\mathrm{d}x$ for some \ac{RV} x) into
    {\allowdisplaybreaks\begin{align*}
        \muvec_{\x|\xbar} &= \E{\x|\xbar} = \E{\xbar + \ntx(\xbar)|\xbar} \\
        &= \E{\xbar|\xbar} + \E{\ntx(\xbar)|\xbar} = \xbar + \muvec_{\ntx|\xbar}.
    \end{align*}}
    Analogously, we proceed for the covariance:
    \begin{align*}
        \covmat_{\x|\xbar} &= \Cov{\x|\xbar} = \Cov{\xbar + \ntx|\xbar} \\
        &= \Cov{\xbar|\xbar} + \Cov{\ntx|\xbar} \\
        &= \nullmatrix + \covmat_{\ntx|\xbar} =  \covmat_{\ntx|\xbar}.\qedhere\vspace*{-3mm}
    \end{align*}%
\end{proof}
\scaleProposition
\begin{resultsinsdcn}[$\x$]
    For both \ac{SIN} and \ac{SDCN}, the \ac{TX} noise is signal-independent and thus we have
    \begin{subequations}
        \begin{equation}
            \muvec_{\x|\xbar} = \xbar + \muvec_{\ntx}\SIN 
        \end{equation}
        \begin{equation}
            \covmat_{\x|\xbar} = \covmat_{\ntx}\SIN.
        \end{equation}
    \end{subequations}
\end{resultsinsdcn}
\scaleProof\begin{proof}
    Apply $\E{\cdot}$ and $\Cov\cdot$ to \eqref{eq:system_model:tx_noise} and then exploit that the \ac{TX} noise is symbol-independent and parametrized by $\muvec_{\ntx}\SIN$ and $\covmat_{\ntx}\SIN$ (see Section~\ref{sec:system_model:special_cases:SIN}).
\end{proof}

\vspace*{-4mm}
\subsubsection{Moments of $\y$}\label{sec:detection:moments:y}
Using the results from Section~\ref{sec:detection:moments:x}, we now derive the moments of $\y$ for a given symbol $\xbar$. The channel noise may depend on the specific realization of $\x$ which is drawn from a distribution itself, rendering the computations more complex than the ones required to obtain $\muvec_{\x|\xbar}$ and $\covmat_{\x|\xbar}$.
\vspace*{-0.5mm}\begin{resultgeneral}[$\y$]
    The mean and covariance of $\y$ for a given symbol $\xbar$ can be expressed as 
    \begin{subequations}\label{eq:detector:moments_y:main}
        \begin{equation}\label{eq:detector:moments_y:main:mean}
            \muvec_{\y|\xbar} = \H \muvec_{\x|\xbar} + \int \nc \int p_{\nc|\x}(\nc|\x) p_{\x|\xbar}(\x|\xbar) \mathrm{d}\x \mathrm{d}\nc
        \end{equation}
        \begin{equation}\label{eq:detector:moments_y:main:covariance}
            \covmat_{\y|\xbar} = \H\covmat_{\x|\xbar}\H\transpose + \int \covfunc(\nc,\muvec_{\y|\xbar})  \int p_{\nc|\x}(\nc|\x) p_{\x|\xbar}(\x|\xbar) \mathrm{d}\x \mathrm{d}\nc.
        \end{equation}
    \end{subequations}
\end{resultgeneral}
\scaleProof\begin{proof}
    Using the linearity of the expectation operator applied to \eqref{eq:system_model:channel_noise} yields
    \vspace*{-5mm}\begin{subequations}
        \begin{equation*}
            \muvec_{\y|\xbar} = \H\muvec_{\x|\xbar} + \E{\nc(\x)|\xbar}
        \end{equation*}
        \begin{equation*}
            \covmat_{\y|\xbar} =  \H\covmat_{\x|\xbar}\H\transpose + \Cov{\nc(\x)|\xbar}. 
        \end{equation*}
    \end{subequations}
    Thus, we see that the mean and the covariance can be divided into a component due to $\x$ and component due to the channel noise $\nc(\x)$. To obtain the mean of $\nc(\x)$ for a given $\xbar$, we use 
    \begin{align*}
        \E{\nc(\x)|\xbar} &= \int \nc p_{\nc|\xbar}(\nc|\xbar) \mathrm{d}\nc \\
        &\stackrel{\text{(a)}}{=} \int \nc \int p_{\x,\nc|\xbar}(\x,\nc|\xbar) \mathrm{d}\x \mathrm{d}\nc \\
        &\stackrel{\text{(b)}}{=}  \int \nc \int p_{\nc|\x}(\nc|\x) p_{\x|\xbar}(\x|\xbar) \mathrm{d}\x \mathrm{d}\nc.
    \end{align*}
    
    \vspace*{-2mm}
    First, for (a), we rewrite the conditional distribution of $\nc$ by marginalizing the joint distribution of $\x$ and $\nc$ (conditional on~$\xbar$) and then for (b), we exploit that $\nc$ and $\xbar$ are conditionally independent given $\x$. %
    Thus, for (b) it follows that the joint distribution can be rewritten in terms of $p_{\nc|\x}(\nc|\x)$ and $p_{\x|\xbar}(\x|\xbar)$, the latter being obtained by convolving $p_{\ntx|\xbar}(\ntx|\xbar)$ with the multi-dimensional generalization of a Dirac impulse at $\xbar$.
    The expression for the covariance matrix $\Cov{\nc(\x)|\xbar}$ is obtained analogously.
\end{proof}
\scaleProposition
\begin{resultsin}[$\y$]
    For the special case of \ac{SIN}, the moments can be straightforwardly obtained as \vspace*{-2mm}\begin{subequations}
        \begin{align}
            \muvec_{\y|\xbar} &= \H\muvec_{\x|\xbar} + \muvec_{\nc}\SIN\\
            \covmat_{\y|\xbar} &= \H \covmat_{\x}\H\transpose + \covmat_{\nc}\SIN.
        \end{align}
    \end{subequations}
\end{resultsin}
\vspace*{-2mm}
\scaleProof\begin{proof}
    We again exploit the linearity of the expectation operator applied to \eqref{eq:system_model:channel_noise} and then insert the definitions for the mean and covariance of the channel noise from Section~\ref{sec:system_model:special_cases:SIN}.
\end{proof}
\scaleProposition
\begin{resultsdcn}[$\y$]
    For \ac{SDCN}, we compute the conditional mean of $\y$ and estimate the conditional covariance of $\y$ using the \ac{UT} instead of computing \eqref{eq:detector:moments_y:main:covariance}, yielding
    \begin{subequations}
        \begin{equation}\label{eq:detection:y:sdcn:mean}
            \muvec_{\y|\xbar} = \H\muvec_{\x|\xbar}
        \end{equation}
        \begin{equation}
            \covmatapprox_{\y|\xbar} = \H \covmat_{\x}\H\transpose + \diagvec{\diagmat{\covmatapprox_{\sqrt{\nuc\H\x}|\xbar}}} + \diagvec{\muvecapprox^2_{\sqrt{\nuc\H\x} | \xbar}},
        \end{equation}
    \end{subequations}
    where $\muvecapprox_{\sqrt{\nuc\H\x} | \xbar}$ and $\covmatapprox_{\sqrt{\nuc\H\x}|\xbar}$ are respectively the approximate mean and covariance of random vector $\sqrt{\nuc\H\x}$. 
\end{resultsdcn}
\scaleProof\begin{proof}
    We consider a slight generalization of \eqref{eq:system_model:SDCN:noise_definition} and assume that $\nc(\x)$ is given as \mbox{$\nc(\x) = \g\x \hadamard \nbase$}, where $\nbase$ is zero-mean Gaussian with covariance matrix $\I$ and $\g\x$ is some non-linear function of $\x$. For \ac{SDCN}, we have $\g{\x} = \sqrt{\nuc\H\x}$.
    Using the auxiliary random vector $\a = \g\x$, we obtain for the mean $\muvec_{\nc(\x)|\xbar} = \E{\a \hadamard \nbase|\xbar} = \E{\a|\xbar} \hadamard \E{\nbase|\xbar} = \E{\a|\xbar} \hadamard \nullmatrix = \nullmatrix,$
        thus yielding \eqref{eq:detection:y:sdcn:mean} directly. 

    The expression for $\Cov{\nc(\x)|\xbar}$ is not as straightforward but can be approximated in two steps: First, we obtain the mean $\muvecapprox_{\a|\xbar}$ and covariance $\covmatapprox_{\a|\xbar}$  of $\a$ using the \ac{UT}. Second, we compute the element in the $i$-th row and the $j$-th column of the covariance matrix of $\u = \a \hadamard \nbase$ according to 
    \vspace*{-2mm}
    \begin{align*}
        \left[ \covmat_{\u|\xbar} \right]_{i,j} &= \left[\E{(\u-\muvec_{\u|\xbar})(\u-\muvec_{\u|\xbar})\transpose|\xbar}\right]_{i,j} \\
        &= \E{[\u-\muvec_{\u|\xbar}]_i[\u-\muvec_{\u|\xbar}]_j|\xbar} \\
        &= \E{[\u-\muvec_{\a|\xbar}\hadamard\nullmatrix]_i[\u-\muvec_{\a|\xbar}\hadamard\nullmatrix]_j|\xbar} \\
        &= \E{[\u]_i[\u]_j|\xbar} = \E{[\a \hadamard \nbase]_i[\a \hadamard \nbase]_j|\xbar} \\
        &= \E{ [\a]_i [\nbase]_i [\a]_j [\nbase]_j |\xbar }  \\
        &\stackrel{\text{(a)}}{=} \E{ [\a]_i [\a]_j|\xbar} \cdot \E{ [\nbase]_i [\nbase]_j |\xbar }  \\
        &\stackrel{\text{(b)}}{=} \begin{cases}
            [\covmatapprox_{\a|\xbar}]_{i,i} + [ \muvecapprox_{\a|\xbar} ]_i^2 & \text{if } i = j \\
            0 & \text{otherwise}
        \end{cases}.
    \end{align*}
    
    \vspace*{-2mm}
    For (a), we used that $\E{h_1(x_1)h_2(x_2)}= \E{h(x_1)} \E{h_2(x_2)}$, for some independent RVs $x_1$, $x_2$ and general functions $h_1(\cdot)$, $h_2(\cdot)$. For (b), we used that for $i = j$, $\E{ [\a]_i [\a]_j|\xbar} = [\covmatapprox_{\a|\xbar}]_{i,i} + [ \muvecapprox_{\a|\xbar} ]_i^2$, and for $i \neq j$ that $\E{ [\nbase]_i [\nbase]_j |\xbar } = 0$. Note that this result is intuitive for our channel noise formulation in Section~\ref{sec:system_model:special_cases:SDCN}, which is essentially an approximation of Poisson noise for $\nuc=1$ and affects each molecule type individually \cite{jamali:olfaction_inspired_MC}, thus rendering the noise affecting different molecule types independent. 
\end{proof}
In summary, we showed in this subsection not only a general approach for obtaining the likelihoods of each symbol and a simple solution for the case of \ac{SIN}, but also how the \ac{UT} can be used to obtain approximate moments of $\nc(\x) = \g\x \hadamard \nbase$. It is worth noting that the same approach could also be generalized to more complex expressions of $\nc(\x)$.

\subsubsection{Moments of $\z$}
As a final step, we obtain the approximate moments of $\z$ for a given symbol $\xbar$ using the \ac{UT} and the results from the previous section. 
\begin{resultgeneral}[$\z$]
We approximate the mean $\muvec_{\z|\xbar}$ and the covariance $\covmat_{\z|\xbar}$ respectively as $\muvecapprox_{\z|\xbar}$ and $\covmatapprox_{\z|\xbar}$, which can be obtained as follows: 
    {\allowdisplaybreaks\begin{subequations}\label{eq:detector:moments_z:main}
        \begin{equation}
            \muvecapprox_{\z|\xbar} = \muvecapprox_{\f\y|\xbar} + \!\!\int \!\!\!\nrx \!\!\!\!\int \!\!\!p_{\nc|\f\y}(\nrx|\f\y) p^{\mathrm{A}}_{\f\y|\xbar}(\f\y|\xbar) \mathrm{d}\y \mathrm{d}\nrx
        \end{equation}
        \begin{equation}
        \begin{split}
            \covmatapprox_{\z|\xbar} &= \covmatapprox_{\f\y|\xbar} +\int \covfunc(\nrx,\muvecapprox_{\z|\xbar})\\ &\cdot \int p_{\nc|\f\y}(\nrx|\f\y) p^{\mathrm{A}}_{\f\y|\xbar}(\f\y|\xbar) \mathrm{d}\y \mathrm{d}\nrx.
            \end{split}
        \end{equation}
    \end{subequations}}
    
    Here, $\muvecapprox_{\f\y|\xbar}$ and $\covmatapprox_{\f\y|\xbar}$ denote the approximate mean and covariance of $\f\y$ that are obtained by propagating $\muvec_{\y|\xbar}$ and $\covmat_{\y|\xbar}$, or their approximations, through $\f\cdot$ using the \ac{UT}. These moments are then used to parametrize the approximate distribution $p^{\mathrm{A}}_{\f\y|\xbar}(\f\y|\xbar)$ of $\f\y$ for a given symbol $\xbar$. 
\end{resultgeneral}
\scaleProof\begin{proof}
    The proof is analogous to Section~\ref{sec:detection:moments:y} except that we employ the estimates $\muvecapprox_{\f\y|\xbar}$ and $\covmatapprox_{\f\y|\xbar}$ and that we have to use an approximate distribution $p^{\mathrm{A}}_{\f\y|\xbar}(\f\y)$ instead of $p_{\x|\xbar}(\x|\xbar)$ due to the non-linearity of $\f\cdot$. 
\end{proof}
\scaleProposition
\begin{resultsinsdcn}[$\z$]
    For both \ac{SIN} and \ac{SDCN}, we obtain
    \vspace*{-3mm}
    \begin{subequations}
        \begin{equation}
            \muvecapprox_{\z|\xbar} = \muvecapprox_{\f\y|\xbar} + \muvec_{\nrx}\SIN
        \end{equation}
        \begin{equation}
            \covmatapprox_{\z|\xbar} = \covmatapprox_{\f\y|\xbar} + \covmat_{\nrx}\SIN.
        \end{equation}
    \end{subequations}
\end{resultsinsdcn}
\scaleProof\begin{proof}
    The proof is analogous to the corresponding proof in Section~\ref{sec:detection:moments:y} with the same adaptations as in the \mbox{previous proof}.
\end{proof}

\vspace*{-6mm}\scaleSection\section{Mixture Alphabet Design}\scaleSectionBelow\label{sec:mixture_design}
In Section~\ref{sec:mixture_design:overview}, we introduce a simple algorithm to generate a mixture alphabet accounting for the characteristics of the \ac{RX} array, and in Section~\ref{sec:mixture_design:metrics}, we discuss different metrics to quantify the pairwise \textit{separability} of symbols in the sensing \mbox{unit output domain}. 

\scaleSubsection\subsection{Algorithm Overview}\scaleSubsectionBelow\label{sec:mixture_design:overview}
While existing schemes for designing signal constellations such as~\cite{wang:effective_constellation_design_CSK} make idealizing assumptions about the \ac{RX} characteristics, e.g., the ability to count individual molecules, or consider only a specific \ac{RX} as in \cite{jamali:olfaction_inspired_MC}, our algorithm is applicable to general \acp{RX} with possibly non-linear and cross-reactive characteristics. As this problem will in general not permit finding globally optimal solutions (in terms of minimizing the \ac{SER} for a given detector) in a reasonable time, we propose a suboptimal, greedy algorithm that optimizes the pair-wise \textit{separability} of the symbols in the sensing unit output domain. 

\vspace*{-3mm}
\begin{algorithm}
\caption{Mixture Alphabet Design Algorithm\label{alg:mixture_design}}
\begin{algorithmic}[1]
\State \textbf{Initialization}
\State \text{Initialize symbol alphabet:} $\symbolalphabet = \{\}$
\State \text{Generate initial point:} $\mathbf{i}_0 \sim \mathcal{U}(\feasibleset)$
\State \text{Generate initial candidate set} $\mathcal{C}$ \text{s.t.} $\mathbf{c} \sim \mathcal{U}(\feasibleset)$, $\forall \mathbf{c} \in \mathcal{C}$
\State \text{Add initial signal point:} $\symbolalphabet = \symbolalphabet \cup\{ \arg\max_{\mathbf{c} \in \mathcal{C}} d(\mathbf{c}, \mathbf{i}_0)$\}

\State \textbf{Iteration}
\While{ $|\symbolalphabet| \leq \nsymbols $ }
    \State \text{Generate candidate set} $\mathcal{C}$ \text{s.t.} $\mathbf{c} \sim \mathcal{U}(\feasibleset)$, $\forall \mathbf{c} \in \mathcal{C}$
    \State \mbox{\text{Add signal point:} $\symbolalphabet = \symbolalphabet \cup\{ \arg\max_{\mathbf{c} \in \mathcal{C}} \;\min_{\mathbf{s} \in \symbolalphabet} d(\mathbf{c}, \mathbf{s})\}$ }
\EndWhile
\end{algorithmic}
\end{algorithm}
\vspace*{-4mm}

Our approach, summarized in Algorithm~\ref{alg:mixture_design}, can be divided into an initialization stage during which the first symbol is chosen and into an iteration stage during which the remaining symbols are added one by one until $\symbolalphabet$ contains the desired number of elements, $\nsymbols$. 
During the \textbf{initialization stage}, we choose an initial point $\mathbf{i}_0$ that is drawn from the uniform distribution over the feasible set $\feasibleset$ (line~3) and, just as during the iteration phase, we create a candidate set $\mathcal{C}$ whose $|\mathcal{C}|$ elements are uniformly distributed in $\feasibleset$ (line~4). Then, we select the candidate with the largest distance to $\mathbf{i}_0$ (as measured by some distance metric $d(\cdot, \cdot)$, which we will introduce in Section~\ref{sec:mixture_design:metrics}) as the first symbol (line~5). This procedure ensures that the first symbol will usually lie somewhere close to the boundary of $\feasibleset$ and not in the middle of $\feasibleset$, which could impair the overall quality of the mixture alphabet. 
Then, during the \textbf{iteration stage} (lines~6-10), we add one new symbol per iteration until the full mixture alphabet is constructed. There, the candidate of the newly generated candidate set $\mathcal{C}$ with the largest minimum distance to any existing symbol in $\symbolalphabet$ is selected.
The quality of the constructed mixture alphabet is determined by the size of the candidate set $|\mathcal{C}|$ and the chosen distance metric. Thus, we propose several distance metrics in the next section.

\scaleSubsection\subsection{Distance Metrics}\scaleSubsectionBelow\label{sec:mixture_design:metrics}
Here, we introduce three distance metrics that quantify the pairwise separability between two symbols in the sensing unit output domain. 
All metrics rely only on the approximate first- and second-order moments in the sensing unit output domain, which were obtained in Section~\ref{sec:detection:moment_derivation}. 

\subsubsection{$\ell_2$ Distance.} An obvious choice for $d(\cdot,\cdot)$ is the $\ell_2$ distance between the mean values of the two symbols $\xbar_i$ and $\xbar_j$ in the sensing unit output domain, i.e.,
\vspace*{-1mm}
 \begin{equation}
    \deuclidean(\xbar_i, \xbar_j) = \lVert \muvecapprox_{\z|\xbar_i} - \muvecapprox_{\z|\xbar_j} \rVert_2.
\end{equation}

\vspace*{-3mm}
\subsubsection{Distance Metric from \cite{jamali:olfaction_inspired_MC}.} To account for signal spread, we use the distance metric incorporating the symbols' covariances that was introduced in \cite{jamali:olfaction_inspired_MC} as a generalization of the \ac{SNR} to \ac{MMSK}. Specifically, this metric increases with the Euclidean distance between the mean values of the two symbols, but it decreases with increasing variance along the shortest path between the mean values. Formally, the metric is given by
\vspace*{-2mm}
\begin{equation}\label{sec:mixture_design:metrics:vahid}
    \dvahid(\xbar_i, \xbar_j) = \frac{ \lVert \muvecapprox_{\z|\xbar_i} - \muvecapprox_{\z_j|\xbar} \rVert_2^2 }{ \mathbf{p}\transpose\covmatapprox_{\z|\xbar_i}\mathbf{p} + \mathbf{p}\transpose\covmatapprox_{\z|\xbar_j}\mathbf{p} }
\end{equation}

\vspace*{-2mm}
where $\mathbf{p} = \frac{\muvecapprox_{\z|\xbar_i} - \muvecapprox_{\z|\xbar_j}}{\lVert \muvecapprox_{\z|\xbar_i} - \muvecapprox_{\z|\xbar_j} \rVert_2}$. While computationally more complex, $\dvahid(\cdot, \cdot)$ promises a higher performance compared to $\deuclidean(\cdot, \cdot)$ due to the consideration of the covariances of both symbols. 

\subsubsection{Approximate \ac{PEP}.} Since the goal of Algorithm~\ref{alg:mixture_design} is to maximize the minimum \ac{PEP} which serves as a proxy for the \ac{SER}, a natural metric is the approximate \ac{PEP} between two symbols:
\vspace*{-2mm}
\begin{equation}
    \dber(\xbar_i, \xbar_j) = - \int \min\left\{ p^{\mathrm{A}}_{\z|\xbar_i}(\z|\xbar_i),\; p^{\mathrm{A}}_{\z|\xbar_j}(\z|\xbar_j) \right\} \mathrm{d}\z,\vspace*{-2mm}
\end{equation}
where $p^{\mathrm{A}}_{\z|\xbar_i}(\z|\xbar_i)$ and $p^{\mathrm{A}}_{\z|\xbar_j}(\z|\xbar_j)$ are the approximate likelihoods for the two symbols $\xbar_i$ and $\xbar_j$ introduced in Section~\ref{sec:detection}. Here, the ''$-$'' is required so that smaller \acp{PEP} correspond to larger distances. Note that $\dber(\cdot, \cdot)$ incurs the highest computational complexity of all proposed metrics as it requires numerical integration.

In summary, we proposed a simple algorithm for mixture alphabet design optimizing the pairwise separability in the sensing unit output domain based on one of the metrics introduced in Section~\ref{sec:mixture_design:metrics}.

\scaleSection\section{Evaluation}\scaleSectionBelow\label{sec:evaluation}
In this section, we present simulation results for the \ac{AML} detector proposed in Section~\ref{sec:detection} and the mixture alphabet design algorithm proposed in Section~\ref{sec:mixture_design} and compare both to relevant baselines.
We provide simulations for both the \ac{SIN} scenario and for the \ac{SDCN} scenario in conjunction with the non-linear model for an array of two \ac{MOS} sensors (see Section~\ref{sec:system_model:special_cases} for the definitions). We parametrized~\eqref{eq:system_model:mos} by fitting it in a least-squares sense to the measurement data for two commercially available \ac{MOS} sensors (TGS800 and TGS826) reported in~\cite{hirobayashi:verification_logarithmic_model_estimation_gas_concentrations} responding to mixtures of ethanol and ammonia, two of the most relevant evolutionary communication molecules. 
The fitted coefficients $\a\sensorone$, $\a\sensortwo$, $\b\sensorone$, $\b\sensortwo$ for the two sensors are reported together with the remaining parameters\footnote{While we consider in this work a molecule type-specific molecule budget, as would be induced by typical odor delivery devices such as the one reported in \cite{hopper:multichannel_portable_odor_delivery_device}, both our detector and our mixture alphabet design algorithm are also compatible with other molecule constraints, which will be explored in future work.}$^,$\footnote{We consider the cases of 8 or 16 mixtures, corresponding to a rate of $3\mathrm{bit}/\mathrm{channel\;use}$ and $4\mathrm{bit}/\mathrm{channel\;use}$, respectively.} that are used for all our simulations in Table~\ref{tab:params}.
To vary the noise power, and thus the \ac{SNR}, relative to the values reported in Table~\ref{tab:params}, we multiply all covariance matrices by a common factor $\nu$ and set $\nuc=\nu$.
\begin{table}[!t]
    \vspace*{-2mm}
    \centering
    \caption{Default system parameters for the simulations.}
    \vspace*{-0.4cm}
    \def\arraystretch{1.2}
    \begin{tabular}{|l|c|c|} %
        \hline
        Parameter & Value (SIN) & Value (SDCN)\\ %
        \hline
        \hline
        $\nsensors$ & \multicolumn{2}{c|}{$2$} \\ \hline
        $\nspecies$ & \multicolumn{2}{c|}{$2$} \\ \hline\hline
        $\H$ & \multicolumn{2}{c|}{$0.01 \cdot \I$} \\ \hline\hline
        $\muvec_{\ntx}$ & \multicolumn{2}{c|}{$\nullmatrix$} \\ \hline
        $\covmat_{\ntx}$ & $10^6 \cdot \I$ & $10^2 \cdot \I$ \\ \hline
        $\muvec_{\nc}$ & $\nullmatrix$ & - \\ \hline
        $\covmat_{\nc}$ & $1 \cdot \I$ & - \\ \hline
        $\muvec_{\nrx}$ & \multicolumn{2}{c|}{$\nullmatrix$} \\ \hline
        $\covmat_{\nrx}$ & $10^{-12} \cdot \I$ & $0.5 \cdot 10^{-12} \cdot \I$ \\ \hline\hline
        $\a\sensorone$ & \multicolumn{2}{c|}{$[2.02 \cdot 10^{-13}, 6.46 \cdot 10^{-6}, 1.61 \cdot 10^{-14}, 7.43 \cdot 10^{-7}]$} \\ \hline
        $\b\sensorone$ & \multicolumn{2}{c|}{$[2.54, 4.67 \cdot 10^{-1}]$} \\ \hline
        $\a\sensortwo$ & \multicolumn{2}{c|}{$[2.16 \cdot 10^{-7}, 2.65 \cdot 10^{-6}, -2.11 \cdot 10^{-9}, 7.77 \cdot 10^{-6}]$}  \\ \hline
        $\b\sensortwo$ & \multicolumn{2}{c|}{$[7.32 \cdot 10^{-1}, 5.122 \cdot 10^{-1}]$} \\ \hline
        \hline
        $\feasibleset$ & \multicolumn{2}{c|}{ $[20 \cdot 10^3, 10 \cdot 10^4] \times [15 \cdot 10^3, 50 \cdot 10^3]$ } \\ \hline
        $|\mathcal{C}|$ & \multicolumn{2}{c|}{ $200$ } \\ 
        \hline
    \end{tabular}
    \label{tab:params}
    \vspace*{-4mm}
\end{table}

\scaleSubsection\subsection{Mixture Alphabet Design Analysis}\scaleSubsectionBelow\label{sec:evaluation:mixture_design}
Before starting with the quantitative evaluation of the mixture alphabets created by our algorithm from Section~\ref{sec:mixture_design}, we first provide some intuition \textbf{why it is important to account for the characteristics of the \ac{RX} array}. 
\begin{figure}
    \centering
    \includegraphics[width=\linewidth]{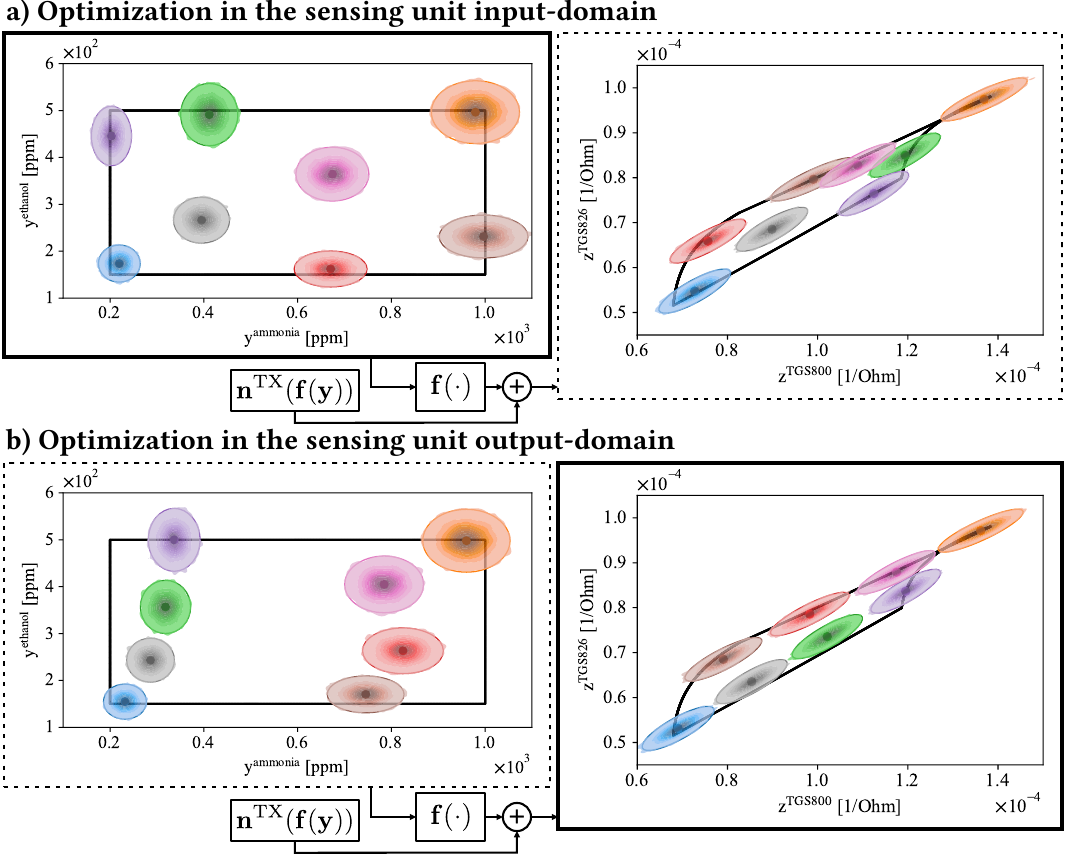}
    \vspace*{-8.2mm}
    \caption{\textbf{Exemplary mixture alphabets for different optimization domains.}}
    \label{fig:eval:optimization_domain_comparison}
    \vspace*{-7mm}
\end{figure}
To this end, we run our mixture alphabet design algorithm with $\dvahid(\cdot,\cdot)$ for the \ac{SDCN} scenario once in the \ac{SOD}, where we optimize for separability of the sensor outputs, and once in the \ac{SID}, where we optimize for separability of the concentrations. Figure~\ref{fig:eval:optimization_domain_comparison} shows the empirical likelihoods of the symbols (shaded areas) of the two resulting alphabets in both the \ac{SID} (left) and the \ac{SOD} (right) together with the $3\sigma$-confidence ellipses according to the moments that are obtained as described in Section~\ref{sec:detection:moment_derivation}. 
Clearly, the optimization in the input domain yields an alphabet with symbols whose concentrations are spread out as far as possible. However, the likelihoods overlap in the output domain due to the non-linear, cross-reactive \ac{RX} behavior. On the other hand, when optimizing for separability in the output domain, the likelihoods are more tightly packed in the input domain but well separable in the output domain, demonstrating the necessity to account for non-linear, cross-reactive \ac{RX} behavior. 

Next, we evaluate the quality of mixture alphabets created by the proposed algorithm with the metrics from Section~\ref{sec:mixture_design:metrics} and compare it with several baselines for $|\symbolalphabet|=8$. To this end, we use the \ac{SER} achieved by the \ac{AML} detector proposed in Section~\ref{sec:detection} as quality metric. 

As a first baseline, we use \textbf{Algorithm~\ref{alg:mixture_design} in the \ac{SID}}, i.e., we compute $\dvahid(\cdot, \cdot)$ for the approximate likelihoods of $\y$ (cf. Section~\ref{sec:detection:moments:y}) instead of $\z$ and thus do not account for the non-linear, cross-reactive \ac{RX} behavior. As a second baseline, we consider \textbf{equally spaced \ac{CSK} for ethanol}\footnote{We chose ethanol as it elicits, compared to ammonia, a larger response range for the individual sensors when choosing signal points equally spaced in $\xbar^{\mathrm{ethanol}} \in [15 \cdot 10^3, 50 \cdot 10^3]$ while fixing $\xbar^{\mathrm{ammonia}}=72 \cdot 10^3 \mathrm{ppm}$.}. This helps to identify the gain that can be achieved by allowing a second dimension for the mixture design. Finally, we also consider a \textbf{random mixture design}, where the individual symbols are simply drawn independently from $\mathcal{U}(\feasibleset)$. 

\begin{figure}
    \centering
    \includegraphics[width=\linewidth]{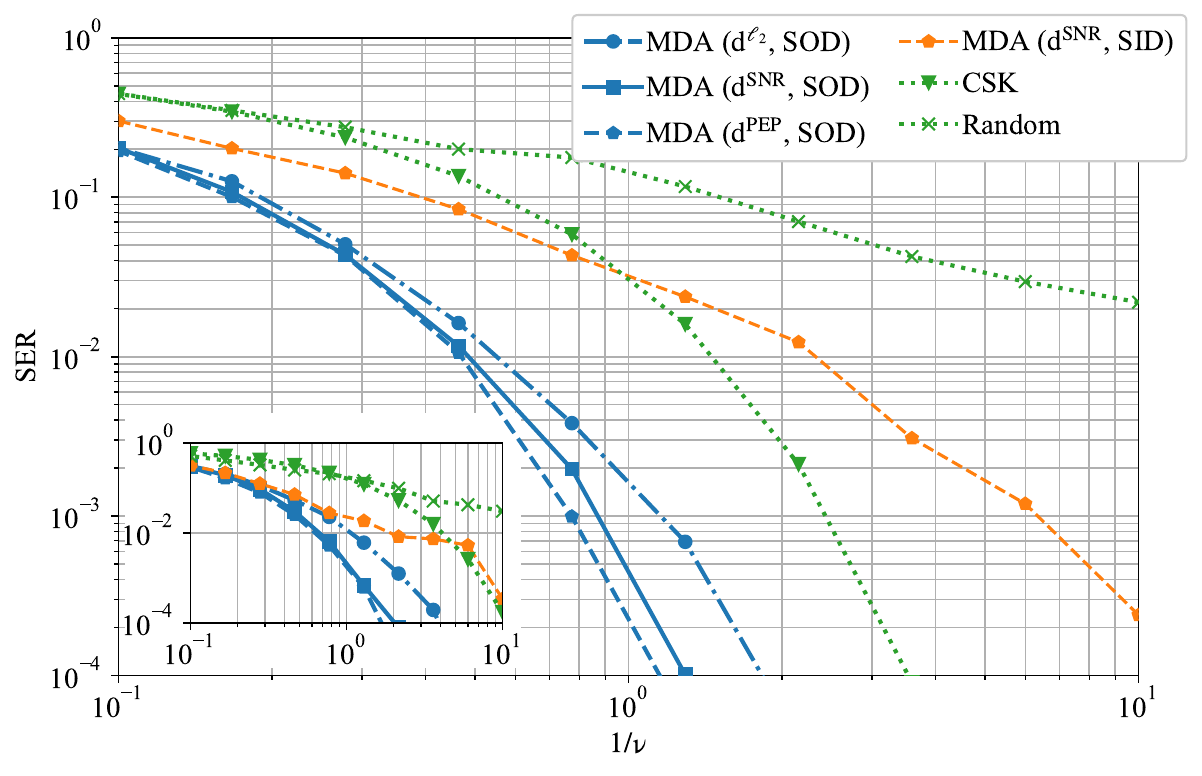}
    \vspace*{-10mm}
    \caption{\textbf{Mixture alphabet design performance analysis.}}
    \label{fig:eval:mixture:sin}
    \vspace*{-6mm}
\end{figure}
In Figure~\ref{fig:eval:mixture:sin}, we compare our \ac{MDA} with the three metrics introduced in Section~\ref{sec:mixture_design:metrics} to the aforementioned baselines for \ac{SIN}\footnote{We plotted the results for \ac{SDCN} as inset in Figure~\ref{fig:eval:mixture:sin} but focus on the \ac{SIN} case for the sake of brevity as we observe the same trends.}. As expected, our \ac{MDA} with $\dber(\cdot,\cdot)$ achieves the lowest \ac{SER} for a given value of $1/\nu$ while the random mixture design has the highest \ac{SER}. However, computing $\dber(\cdot,\cdot)$ requires the numerical computation of the approximate \acp{PEP} and thus incurs much higher computational cost compared to our approach using $\dvahid(\cdot,\cdot)$, which yields similar performance. Using $\deuclidean(\cdot,\cdot)$ has even lower computational complexity, but also a larger relative performance loss.
Regardless of the employed distance metric, our algorithm significantly outperforms all baselines, including the mixture design in the \ac{SID} (dashed, orange) which does not account for the non-linear, cross-reactive \ac{RX} behavior. 
In this specific setting, the alphabet design in the \ac{SID} is even outperformed by \ac{CSK}, whose signal elements are positioned so that $\f\cdot$ does not cause additional overlaps between the likelihoods in the \ac{SOD}, in contrast to the optimization in the \ac{SID} (see Fig.~\ref{fig:eval:optimization_domain_comparison}).

\scaleSubsection\subsection{Detector Analysis}\scaleSubsectionBelow\label{sec:evaluation:detector}
In the following, we focus on the performance of the \ac{AML} detector proposed in Section~\ref{sec:detection}. To this end, we compare, for given mixture alphabets, our detector to several baselines, namely the \textit{centroid detector}, \textit{\ac{kNN} classifiers}, and a \textit{histogram-based approach}, all of which we describe in the following in more detail.

The \textbf{centroid detector} is a \ac{UT}-based detector but with $\covmatapprox_{\z|\xbar}=\I$, $\forall\xbar$, i.e., only the mean - or centroid - of the likelihoods is considered. While this introduces some sub-optimality compared to the \ac{UT}-based detector which considers the spread of the individual symbols, it has an even lower computational complexity during the \textit{online phase} as only the computation of $\ell_2$ distances is necessary. Thus, the centroid detector may be suitable for scenarios with very strict hardware and power requirements at the \ac{RX}. 

Furthermore, we consider \textbf{\ac{kNN} classifiers}, which are a commonly used \textit{model-free} approach in machine learning, where new samples are classified according to the classes of the $k$ nearest samples in terms of the Euclidean distance. While this approach has the advantage of not requiring any model of the \ac{TX}, channel, and \ac{RX}, and instead uses real-world measurements (or simulations), its performance depends on the number of training samples whose generation might be costly if they are obtained from measurements. Furthermore, the complexity during the online phase also increases with the number of training samples, yielding a non-trivial accuracy-complexity trade-off. Thus, we consider \ac{kNN} classifiers trained with four samples per symbol, i.e., the same number of evaluations of $\f\cdot$ during the offline phase as the \ac{UT}-based detector, and with 100 samples per symbol, respectively.

Finally, we employ a \textbf{histogram-based \ac{AML} classifier}, where we simulate $10^6$ samples per symbol and compute the corresponding 2D histogram in the sensing unit output domain with a bin width of $10^{-6}$. Then, we classify samples during inference based on the histograms which approximate the likelihoods. While this approach can come arbitrarily close to the optimum by increasing the number of samples used to obtain the histograms, it also has the highest complexity during the offline phase. This becomes especially clear when comparing this approach with $10^6$ training samples per symbol to the \ac{UT}-based detector, which requires only 4 carefully chosen samples per symbol for $\nspecies=2$. 

\begin{figure}
    \vspace*{-4mm}
    \centering
    \includegraphics[width=\linewidth]{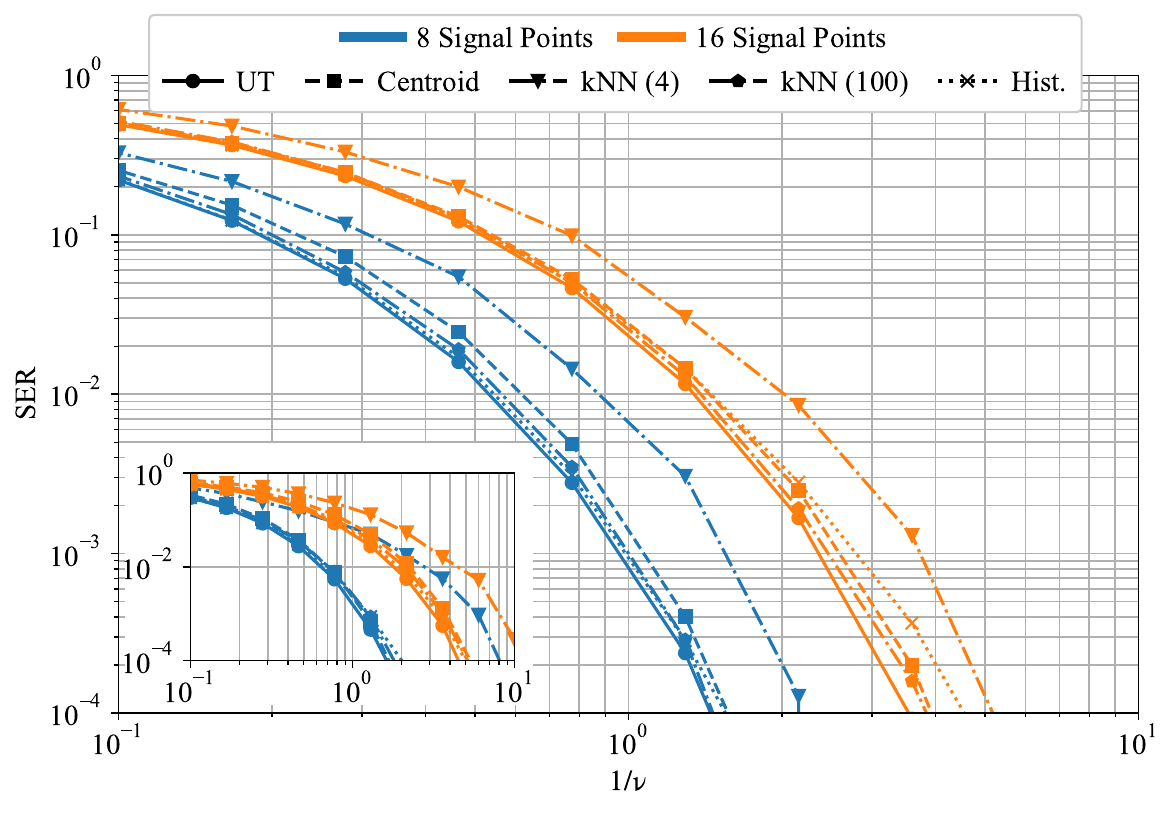}
    \vspace*{-9mm}
    \caption{\textbf{Detector performance analysis.}}
    \label{fig:eval:detector:sin_8_16}
    \vspace*{-6mm}
\end{figure}

Figure~\ref{fig:eval:detector:sin_8_16} shows the \ac{SER} achieved by the \ac{UT}-based \ac{AML} detector and the aforementioned baselines for 8 and 16 signal points for different values of $1/\nu$ for given mixture alphabets in the \ac{SIN} scenario\footnote{We plot the results for the \ac{SDCN} scenario as inset to Figure~2. As the same trends can be observed for both \ac{SIN} and \ac{SDCN}, we focus on the \ac{SIN} case for brevity.}. 
First, we observe that the \ac{SER} decreases for all shown detectors for increasing $1/\nu$. Our proposed \ac{AML} detector based on the \ac{UT} achieves for all $\nu$ the lowest \ac{SER} although the \ac{kNN} detector with 100 training samples per symbol and the histogram-based detector exhibit almost the same performance at the cost of relying on much more samples during the offline phase. The \ac{kNN} detector with only four training samples per symbol has significantly worse performance than all other schemes, demonstrating that a model-based approach like our \ac{UT}-based detector can achieve significantly higher \textit{sample efficiency} compared to purely data-driven approaches, i.e., a lower number of training samples can achieve a desired performance. As expected, the centroid detector has some performance loss compared to the \ac{UT}-based detector since it neglects the spread of the likelihoods.

In summary, Figure~\ref{fig:eval:detector:sin_8_16} illustrates that the \ac{UT}-based \ac{AML} detector achieves the best \ac{SER} performance compared to several baseline methods, which confirms its applicability in challenging settings with significant non-linear and cross-reactive \ac{RX} characteristics. At the same time, the \ac{UT}-based \ac{AML} detector requires, in comparison to the data-driven baseline methods, significantly fewer training samples for a given performance.

\scaleSection\section{Conclusion}\scaleSectionBelow\label{sec:conclusion}
In this work, we studied the detection and alphabet design for non-linear, cross-reactive \ac{RX} arrays for molecular mixture communication. We proposed an \ac{AML} detector operating on the outputs of the sensor array, as well as a mixture alphabet design algorithm that optimizes symbol separability by accounting for the \ac{RX} behavior. Our proposed detector performs similar to data-driven baseline schemes in terms of \ac{SER} while requiring significantly fewer training samples. Similarly, our mixture alphabet design algorithm improves upon schemes that ignore the non-linear, cross-reactive characteristics of \ac{RX} arrays.
These results highlight the importance of explicitly accounting for the \ac{RX} characteristics for realizing practical \ac{MC} systems. Interesting directions for further research are the experimental validation of the proposed schemes and extending them to handle the uncertainty about the parameters of the system model. 
}

\scaleSection\begin{acks}
\vspace*{-1mm}
This work was funded by the Deutsche Forschungsgemeinschaft (DFG, German Research Foundation) – GRK 2950 – Project-ID 509922606.
\end{acks}

\scaleSection
\bibliographystyle{ACM-Reference-Format}
\bibliography{sample-base}

\end{document}